\documentclass[submission,copyright,creativecommons]{eptcs}
\providecommand{\event}{QPL 2020}
\pdfoutput=1

\usepackage{mathtools} 
\usepackage{amstext} 
\usepackage{amssymb} 
\usepackage{bbold} 
\usepackage{amsthm} 
\usepackage{stmaryrd} 

\usepackage[utf8]{inputenc} 

\usepackage{relsize} 
\usepackage{microtype} 
\usepackage{csquotes} 
\usepackage{ragged2e} 

\usepackage{hyperref} 
\usepackage[compress]{cite} 

\usepackage{graphicx} 
\usepackage[usenames,dvipsnames]{xcolor} 
\usepackage{stackengine}
\usepackage{tikz} 
\usepackage{circuitikz} 
\usetikzlibrary{
	arrows,
	shapes,
	decorations,
    patterns,
	intersections,
	backgrounds,
	positioning,
	circuits.ee.IEC
	}

		\newcounter{theorem_c} 

		\newtheoremstyle{thmstyle}
		  {2mm} 
		  {2mm} 
		  {\itshape} 
		  {} 
		  {\bfseries} 
		  {.} 
		  {.5em} 
		  {} 

		\theoremstyle{thmstyle}
		
		\newtheorem{proposition}[theorem_c]{Proposition}

		\newtheorem{definition}[theorem_c]{Definition}

		\newtheoremstyle{exampstyle}
		  {2mm} 
		  {2mm} 
		  {} 
		  {} 
		  {\bfseries} 
		  {.} 
		  {.5em} 
		  {} 

		\theoremstyle{exampstyle}
		\newtheorem{example}[theorem_c]{Example}
		\newtheorem{remark}[theorem_c]{Remark}

\newcommand{\reals}{\mathbb{R}}
\newcommand{\complexs}{\mathbb{C}}

\newcommand{\suchthat}[2]{\left\{#1 \: \middle\vert \: #2\right\}}

\newcommand{\restrict}[2]{\left. #1 \right\vert_{#2}}

\newcommand{\ket}[1]{\left| #1 \right\rangle}
\newcommand{\bra}[1]{\left\langle #1 \right|}

\newcommand{\id}[1]{\operatorname{id}_{#1}}

\newcommand{\obj}[1]{\operatorname{obj}\left(#1\right)}
\newcommand{\mor}[1]{\operatorname{mor}\left(#1\right)}

		\newcommand{\RMatCategory}[1]{#1\operatorname{-Mat}} 
		\newcommand{\CategoryC}{\mathcal{C}}

\tikzset{
  rectangle with rounded corners north west/.initial=4pt,
  rectangle with rounded corners south west/.initial=4pt,
  rectangle with rounded corners north east/.initial=4pt,
  rectangle with rounded corners south east/.initial=4pt,
}
\makeatletter
\pgfdeclareshape{rectangle with rounded corners}{
  \inheritsavedanchors[from=rectangle] 
  \inheritanchorborder[from=rectangle]
  \inheritanchor[from=rectangle]{center}
  \inheritanchor[from=rectangle]{north}
  \inheritanchor[from=rectangle]{south}
  \inheritanchor[from=rectangle]{west}
  \inheritanchor[from=rectangle]{east}
  \inheritanchor[from=rectangle]{north east}
  \inheritanchor[from=rectangle]{south east}
  \inheritanchor[from=rectangle]{north west}
  \inheritanchor[from=rectangle]{south west}
  \backgroundpath{
    \southwest \pgf@xa=\pgf@x \pgf@ya=\pgf@y
    \northeast \pgf@xb=\pgf@x \pgf@yb=\pgf@y
    \pgfkeysgetvalue{/tikz/rectangle with rounded corners north west}{\pgf@rectc}
    \pgfsetcornersarced{\pgfpoint{\pgf@rectc}{\pgf@rectc}}
    \pgfpathmoveto{\pgfpoint{\pgf@xa}{\pgf@ya}}
    \pgfpathlineto{\pgfpoint{\pgf@xa}{\pgf@yb}}
    \pgfkeysgetvalue{/tikz/rectangle with rounded corners north east}{\pgf@rectc}
    \pgfsetcornersarced{\pgfpoint{\pgf@rectc}{\pgf@rectc}}
    \pgfpathlineto{\pgfpoint{\pgf@xb}{\pgf@yb}}
    \pgfkeysgetvalue{/tikz/rectangle with rounded corners south east}{\pgf@rectc}
    \pgfsetcornersarced{\pgfpoint{\pgf@rectc}{\pgf@rectc}}
    \pgfpathlineto{\pgfpoint{\pgf@xb}{\pgf@ya}}
    \pgfkeysgetvalue{/tikz/rectangle with rounded corners south west}{\pgf@rectc}
    \pgfsetcornersarced{\pgfpoint{\pgf@rectc}{\pgf@rectc}}
    \pgfpathclose
 }
}
\makeatother

\tikzstyle{inline text}=[text height=1.5ex, text depth=0.25ex, yshift=0.5mm]
\tikzstyle{upground}=[circuit ee IEC, thick, ground, rotate=90, scale=2]
\tikzstyle{downground}=[circuit ee IEC, thick, ground, rotate=-90, scale=1.5]
\tikzstyle{point}=[regular polygon, regular polygon sides=3, draw, scale=0.75, inner sep=-0.5pt, minimum width=9mm, fill=white, regular polygon rotate=180, tikzit fill={rgb,255: red,242; green,255; blue,92}]
\tikzstyle{wide copoint}=[fill=white, draw, shape=isosceles triangle, shape border rotate=90, isosceles triangle stretches=true, inner sep=0pt, minimum width=1.5cm, minimum height=6.12mm]
\tikzstyle{wide point}=[fill=white, draw, shape=isosceles triangle, shape border rotate=-90, isosceles triangle stretches=true, inner sep=0pt, minimum width=1.5cm, minimum height=6.12mm, yshift=-0.0mm]
\tikzstyle{wide dpoint}=[wide point, doubled]
\tikzstyle{copoint}=[regular polygon, regular polygon sides=3, draw, scale=0.75, inner sep=-0.5pt, minimum width=9mm, fill=white, tikzit fill={rgb,255: red,255; green,128; blue,0}, tikzit draw={rgb,255: red,255; green,128; blue,0}]
\tikzstyle{dot}=[inner sep=0mm, minimum width=2mm, minimum height=2mm, draw, shape=circle]
\tikzstyle{black dot}=[dot, fill={gray!30}, text depth=-0.2mm]
\tikzstyle{white dot}=[dot, fill=white, text depth=-0.2mm]
\tikzstyle{small box}=[rectangle, inline text, fill=white, draw, minimum height=5mm, yshift=-0.5mm, minimum width=5mm, font={\small}]
\tikzstyle{small gray box}=[small box, fill={gray!30}]
\tikzstyle{medium box}=[rectangle, inline text, fill=white, draw, minimum height=5mm, yshift=-0.5mm, minimum width=10mm, font={\small}]
\tikzstyle{square box}=[small box]
\tikzstyle{medium gray box}=[small box, fill={gray!30}]
\tikzstyle{semilarge box}=[rectangle, inline text, fill=white, draw, minimum height=5mm, yshift=-0.5mm, minimum width=12.5mm, font={\small}]
\tikzstyle{large box}=[rectangle, inline text, fill=white, draw, minimum height=5mm, yshift=-0.5mm, minimum width=15mm, font={\small}]
\tikzstyle{large gray box}=[small box, fill={gray!30}]
\tikzstyle{dpoint}=[point, doubled]
\tikzstyle{dcopoint}=[copoint, doubled]
\tikzstyle{boldedge}=[doubled, shorten <=-0.17mm, shorten >=-0.17mm]
\tikzstyle{normal}=[line width=0.9pt]
\tikzstyle{doubled}=[line width=1pt]
\tikzstyle{boldedge}=[doubled, shorten <=-0.17mm, shorten >=-0.17mm]
\tikzstyle{small dbox}=[small box, doubled]
\tikzstyle{white ddot}=[white dot, doubled]
\tikzstyle{black ddot}=[black dot, doubled, tikzit fill=black]
\tikzstyle{map}=[draw, shape=NEbox, inner sep=2pt, minimum height=6mm, fill=white]
\tikzstyle{box}=[draw, shape=rectangle, inner sep=2pt, minimum height=6mm, minimum width=6mm, fill=white]
\tikzstyle{dbox}=[draw, doubled, shape=rectangle, inner sep=2pt, minimum height=6mm, minimum width=6mm, fill=white]
\tikzstyle{dmap}=[draw, doubled, shape=NEbox, inner sep=2pt, minimum height=6mm, fill=white]
\tikzstyle{dmapdag}=[draw, doubled, shape=SEbox, inner sep=2pt, minimum height=6mm, fill=white]
\tikzstyle{dmapadj}=[draw, doubled, shape=SEbox, inner sep=2pt, minimum height=6mm, fill=white]
\tikzstyle{dmaptrans}=[draw, doubled, shape=SWbox, inner sep=2pt, minimum height=6mm, fill=white]
\tikzstyle{dmapconj}=[draw, doubled, shape=NWbox, inner sep=2pt, minimum height=6mm, fill=white]
\tikzstyle{map}=[draw, shape=NEbox, inner sep=2pt, minimum height=6mm, fill=white]
\tikzstyle{dashedmap}=[draw, dashed, shape=NEbox, inner sep=2pt, minimum height=6mm, fill=white]
\tikzstyle{mapdag}=[draw, shape=SEbox, inner sep=2pt, minimum height=6mm, fill=white]
\tikzstyle{mapadj}=[draw, shape=SEbox, inner sep=2pt, minimum height=6mm, fill=white]
\tikzstyle{maptrans}=[draw, shape=SWbox, inner sep=2pt, minimum height=6mm, fill=white]
\tikzstyle{mapconj}=[draw, shape=NWbox, inner sep=2pt, minimum height=6mm, fill=white]
\tikzstyle{semilarge map}=[draw, shape=NEbox, inner sep=2pt, minimum height=6mm, fill=white, minimum width=9.5mm]
\tikzstyle{semilarge dmap}=[draw, doubled, shape=NEbox, inner sep=2pt, minimum height=6mm, fill=white, minimum width=9.5mm]
\tikzstyle{kpointdag}=[kpoint adjoint]
\tikzstyle{kpointadj}=[kpoint adjoint]
\tikzstyle{kpointconj}=[kpoint conjugate]
\tikzstyle{kpointtrans}=[kpoint transpose]
\tikzstyle{kpoint common}=[draw, fill=white, inner sep=1pt, minimum height=4mm]
\tikzstyle{kpoint sc}=[shape=cornerpoint, kpoint common]
\tikzstyle{kpoint adjoint sc}=[shape=cornercopoint, kpoint common]
\tikzstyle{kpoint}=[shape=cornerpoint, shorten left=5pt, kpoint common, tikzit fill={rgb,255: red,255; green,128; blue,0}]
\tikzstyle{kpoint adjoint}=[shape=cornercopoint, shorten left=5pt, kpoint common, tikzit fill={rgb,255: red,255; green,128; blue,0}]
\tikzstyle{kpoint conjugate}=[shape=cornerpoint, shorten right=5pt, kpoint common]
\tikzstyle{kpoint transpose}=[shape=cornercopoint, shorten right=5pt, kpoint common]
\tikzstyle{kpoint symm}=[shape=cornerpoint, shorten left=5pt, shorten right=5pt, kpoint common]
\tikzstyle{wide kpoint}=[kpoint, minimum width=1 cm, inner sep=2pt]
\tikzstyle{wide kpointdag}=[kpointdag, minimum width=1 cm, inner sep=2pt]
\tikzstyle{wide kpointconj}=[kpointconj, minimum width=1 cm, inner sep=2pt]
\tikzstyle{wide kpointtrans}=[kpointtrans, minimum width=1 cm, inner sep=2pt]
\tikzstyle{wider kpoint}=[kpoint, minimum width=1.25 cm, inner sep=2pt]
\tikzstyle{wider kpointdag}=[kpointdag, minimum width=1.25 cm, inner sep=2pt]
\tikzstyle{wider kpointconj}=[kpointconj, minimum width=1.25 cm, inner sep=2pt]
\tikzstyle{wider kpointtrans}=[kpointtrans, minimum width=1.25 cm, inner sep=2pt]
\tikzstyle{dkpoint}=[kpoint, doubled, tikzit fill={rgb,255: red,255; green,85; blue,210}]
\tikzstyle{wide dkpoint}=[wide kpoint, doubled, tikzit fill={rgb,255: red,68; green,255; blue,0}]
\tikzstyle{dkpointdag}=[kpoint adjoint, doubled]
\tikzstyle{wide dkpointdag}=[wide kpointdag, doubled]
\tikzstyle{label}=[fill=white, draw=white, shape=circle, tikzit draw={rgb,255: red,10; green,26; blue,255}, tikzit fill={rgb,255: red,0; green,12; blue,255}, font={\small}]
\tikzstyle{squarelabel}=[fill=white, draw=white, shape=rectangle, tikzit draw=black]
\tikzstyle{eslabel}=[tikzit draw={rgb,255: red,255; green,191; blue,191}, tikzit fill={rgb,255: red,255; green,191; blue,191}, font={\tiny}]
\tikzstyle{large dmap}=[draw, doubled, shape=NEbox, inner sep=2pt, minimum height=6mm, fill=white, minimum width=12mm]
\tikzstyle{gray point}=[point, fill={gray!40!white}]
\tikzstyle{gray dpoint}=[gray point, doubled, tikzit draw={rgb,255: red,128; green,128; blue,128}, tikzit fill={rgb,255: red,128; green,128; blue,128}]
\tikzstyle{gray copoint}=[copoint, fill={gray!40!white}, tikzit fill={rgb,255: red,128; green,128; blue,128}]
\tikzstyle{gray dcopoint}=[gray copoint, doubled, tikzit fill={rgb,255: red,128; green,128; blue,128}]
\tikzstyle{circlenew}=[draw=black, shape=circle, inner sep=1pt]
\tikzstyle{blue label}=[text=NavyBlue, tikzit draw={rgb,255: red,0; green,96; blue,167}, tikzit fill={rgb,255: red,35; green,68; blue,255}]

\tikzstyle{new edge style 1}=[-, line width=1pt, shorten <=-0.17mm, shorten >=-0.17mm, tikzit draw={rgb,255: red,204; green,0; blue,3}]
\tikzstyle{diredge}=[-, postaction=decorate, decoration={markings, mark=at position 0.55 with \edgearrow}]
\tikzstyle{bold diredge}=[-, diredge, line width=1pt, tikzit draw={rgb,255: red,128; green,0; blue,128}]
\tikzstyle{grey}=[-, draw={rgb,255: red,188; green,188; blue,188}]
\tikzstyle{classical}=[-, dashed, tikzit draw={rgb,255: red,255; green,128; blue,0}]
\tikzstyle{reddashed}=[-, dashed, draw={rgb,255: red,0; green,128; blue,128}, postaction=decorate, decoration={markings, mark=at position 0.55 with \edgearrow}]
\tikzstyle{reddahednoarrow}=[-, dashed, draw={rgb,255: red,179; green,40; blue,40}]
\tikzstyle{arrow edge}=[-, ->, draw={rgb,255: red,191; green,191; blue,191}, tikzit draw={rgb,255: red,191; green,191; blue,191}, ultra thick]
\tikzstyle{tarrow edge}=[-, ->, draw={rgb,255: red,191; green,191; blue,191}, tikzit draw={rgb,255: red,191; green,191; blue,191}]
\tikzstyle{gray edge}=[-, draw={rgb,255: red,191; green,191; blue,191}, tikzit draw={rgb,255: red,191; green,191; blue,191}, ultra thick]
\tikzstyle{lightgrayedge}=[-, draw={rgb,255: red,207; green,207; blue,207}]
\tikzstyle{green edge}=[-, tikzit draw={rgb,255: red,128; green,128; blue,0}, draw={rgb,255: red,128; green,128; blue,0}]
\tikzstyle{red edge}=[-, draw={rgb,255: red,191; green,0; blue,64}, tikzit draw={rgb,255: red,191; green,0; blue,64}]
\tikzstyle{arrow edge black}=[-, ->]
\tikzstyle{solid blue}=[-, draw={rgb,255: red,0; green,96; blue,167}, tikzit draw={rgb,255: red,0; green,96; blue,167}]
\tikzstyle{classical blue}=[-, draw={rgb,255: red,0; green,96; blue,167}, tikzit draw={rgb,255: red,0; green,96; blue,167}, dashed]

\usepackage{tikzit}
\usepackage{newtxtext}
\usepackage{newtxmath}
\usepackage{bbm}
\usepackage{dsfont}
\usepackage{relsize}

\newcommand{\edgearrow}{{\arrow[black]{>}}}
\DeclareMathOperator{\Tr}{Tr}

\newcommand{\myfunction}[2]{#1\left(#2\right)}
\newcommand{\inputs}[1]{\myfunction{\operatorname{in}}{#1}}
\newcommand{\outputs}[1]{\myfunction{\operatorname{out}}{#1}}

\newcommand{\nodes}[1]{\myfunction{\operatorname{nodes}}{#1}}
\newcommand{\edges}[1]{\myfunction{\operatorname{edges}}{#1}}
\newcommand{\events}[1]{\myfunction{\operatorname{ev}}{#1}}
\newcommand{\tailSym}{\operatorname{tail}}
\newcommand{\headSym}{\operatorname{head}}
\newcommand{\tail}[1]{\myfunction{\tailSym}{#1}}
\newcommand{\head}[1]{\myfunction{\headSym}{#1}}
\newcommand{\sysmapSym}{\operatorname{sys}}
\newcommand{\procmapSym}{\operatorname{proc}}
\newcommand{\sysmap}[1]{\myfunction{\sysmapSym}{#1}}
\newcommand{\procmap}[1]{\myfunction{\procmapSym}{#1}}
\newcommand{\diagramproc}[1]{\llbracket #1 \rrbracket}
\newcommand{\syslabelSym}[2]{\operatorname{syslabel}_{#1, #2}}
\newcommand{\syslabel}[3]{\myfunction{\syslabelSym{#1}{#2}}{#3}}
\newcommand{\definiteCausAssoc}[1]{\myfunction{\operatorname{DCaus}}{#1}}

\definecolor{burntorange}{rgb}{0.8, 0.33, 0.0} 

\usepackage{csquotes}

\title{Giving Operational Meaning to\\the Superposition of Causal Orders}
\author{
    Nicola Pinzani
    \institute{University of Oxford}
    \email{nicola.pinzani@cs.ox.ac.uk}
    \and
    Stefano Gogioso
    \institute{University of Oxford}
    \email{stefano.gogioso@cs.ox.ac.uk}
}
\def\titlerunning{Giving Operational Meaning to the Superposition of Causal Orders}
\def\authorrunning{N. Pinzani and S. Gogioso}

\begin{document}

\maketitle

\begin{abstract}
    In this work, we give rigorous operational meaning to superposition of causal orders. This fits within a recent effort to understand how the standard operational perspective on quantum theory could be extended to include indefinite causality. The mainstream view, that of ``process matrices", takes a top-down approach to the problem, considering all causal correlations that are compatible with local quantum experiments. Conversely, we pursue a bottom-up approach, investigating how the concept of indefiniteness emerges from specific characteristics of generic operational theories. Specifically, we pin down the operational phenomenology of the notion of non-classical (e.g. ``coherent") control, which we then use to formalise a theory-independent notion of control (e.g. ``superposition") of causal orders. To validate our framework, we show how salient examples from the literature can be captured in our framework.
\end{abstract}

\section{Introduction}

When modelling the interaction of localised processes (aka ``operations'' or ``experiments'') in spacetime, the causal structure between the corresponding events can be captured by circuit-like diagrams.
Such interactions arise as combinations of parallel and sequential compositions of processes, relating to space-like and time-like separation of events respectively: if the processes themselves respect a notion of \emph{causality}---defined as the impossibility of transmitting information from the future to the past---then it can be shown that causal relations between the events at which the processes take place can unambiguously be described by a directed acyclic graph, each edge capturing the forward-in-time information flow from the output of one process to the input of another \cite{Coecke2012,Pinzani2019,DAriano2018}.
For example, taking quantum instruments as processes/operations respects such a notion of causality, but adding post-selection breaks it, resulting in the possibility of signalling from the future.

A 1977 result by Malament \cite{Malament1977} shows that---in the continuous limit---little more than the above would be needed to reconstruct the structure of spacetime itself: knowledge of the causal relationship between its events is enough to reconstruct a spacetime up to conformal equivalence, as long as that spacetime satisfies some mild requirement.
\footnote{Specifically, the requirement is for the spacetime to be both past-distinguishing and future-distinguishing: two events are the same if the have the same causal future or the same causal past. This excludes a number of causally problematic scenarios, such as the existence of closed time-like curves.}
An extension of this identification of causal structure with directed graphs was recently used by the authors to study operational models of quantum information in the presence of causal anomalies, such as closed time-like curves \cite{Pinzani2019}.

In this work, we examine another scenario in which the operational approach to quantum theory needs to be extended in order to cope with exotic causal structure: that where the background spacetime is itself allowed to be in a superposition.
In doing so, we also dispel some of the fog surrounding the notion of \emph{control of causal order}: this is the situation in which the outcome of one process may influence the causal order of subsequent processes, both classically (as a probabilistic mixture) and coherently (as a superposition).

Reasoning about operational theories in the context where spacetime itself becomes a dynamical variable carries a number of additional complications.
For example, thinking about the set-up and outcomes of an experiment presupposes the existence of causally stable surroundings, where the notions of cause and consequence take their familiar form independently of the specific processes being performed.
Failing these assumptions, how can we make sure that a mathematical model of quantum theory in the presence of dynamical spacetimes is even empirically testable?
This is an important question, upon which many others stumbled before us.
For example, the following reflection can be found in a prominent piece of literature on the application of sheaves and topoi to quantum theory \cite{Isham2011,Doering2010}:

\begin{displayquote}[Chris J. Isham \cite{Isham2011}]
    \emph{
    ``[A]round fifteen years ago, I came to the conclusion that the use of standard quantum theory was fundamentally inconsistent, and I stopped working in quantum gravity proper [...] [W]hat could it mean to `measure' properties of space or time if the very act of measurement requires a spatiotemporal background within which it is made?''
    }
\end{displayquote}

\noindent
In this work, we set out to endow superposition of causal orders with rigorous operational meaning, through the development of suitable categorical semantics.

We start by considering a probabilistic theory of processes, modelling the operational aspects of some physical theory (such as quantum theory).
Given such a theory, we want to provide a sound way of constructing and characterising localised processes which don't take place against a fixed causal background, but rather against a ``superposition'' of causal backgrounds, determined by some ``wave-function'' over the set of all fixed causal backgrounds compatible with the processes in question.
We do so by first defining a theory-independent notion of ``control'' of processes, accommodating both the classical case---where the choice of process to execute can be captured by some hidden variable---and the coherent case---where it is not possible to establish which one process was executed without reference to a specific measurement context.
Armed with such a notion, we show how to construct superpositions of diagrams, giving categorical semantics to the execution of processes and operations against a dynamic causal background.

A significant body of literature exists which investigates the informational advantage of superposition of channels and causal orders \cite{oi2003,abbott2018,Chiribella2019,Araujo2014,salek2018,Oreshkov2012} and the possibility of experimentally detecting such superpositions \cite{Marletto2017,Marletto2019,Rovelli2019,Bruckner2019}.
In particular, there has been some recent speculation about the possibility of realising a genuine ``quantum switch'' by coherently controlling the `spatial degree of freedom' \cite{paunkovic2019causal}.
Unfortunately, however, the issue of delimiting a tight operational setting in which to interpret the results is not ordinarily viewed as a necessity, sometimes leading to misinterpretation of their physical significance.
The goal of this paper is then to provide a rigorous standpoint from which to discuss the operational phenomenology of causal superposition and control.
If quantum gravity is ever to be observed, we must first be in possession of the mathematical tools to rigorously understand, debate and communicate those observations.

The categorical semantics in this work are given within the process-theoretic framework for probabilistic theories introduced by \cite{Gogioso2018}, which is briefly summarised in Appendix~\ref{appendix:cpt}.
The same semantics can easily be adapted to other popular frameworks, such as that of Operational Probabilistic Theories (OPTs) \cite{chiribella2010,chiribella2017}.
The main differences between the framework presented here and that of OPTs are as follows:

\begin{enumerate}
    \item[(i)] normalisation is presented constructively, by choice of discarding map and an equation defining the normalised (aka ``deterministic'') processes;
    \item[(ii)] the whole convex cone of processes is considered, including super-normalised components which can be used to decompose processes into simpler building pieces (e.g. as in the ZX calculus)\cite{CD2};
    \item[(iii)] classical systems are explicitly considered as systems within the framework, e.g. allowing for diagrammatic treatment of the quantum-classical interface;
    \item[(iv)] the framework straightforwardly extends to semirings other than the probabilistic semiring $\reals^+$, allowing for discussion of quasi-probabilistic theories (semiring $\reals$), possibilistic theories (boolean semiring) as well as more exotic examples including modal, hyperbolic or $p$-adic quantum theory \cite{gogioso2017}.
\end{enumerate}

\noindent
Despite the presentational differences and the added flexibility, the framework used here contains all causal OPTs as special cases, allowing for all results to be transferred straightforwardly (as long as only normalised and sub-normalised maps are involved).

\section{Controlled processes}

As the first step towards our formulation of semantics for superposition of causal orders, we define a general notion of control of morphisms valid in arbitrary probabilistic theories.
Our definition captures the idea of an agent being able to control the choice of morphisms by suitably encoding and/or decoding classical information about their choice into and/or from a physical system.

\begin{definition}
    Let $\mathcal{C}$ be a probabilistic theory, let $A, B \in \obj{\mathcal{C}}$ be any two systems in the theory.
    Let $(F_x)_{x \in X}$ be a family of processes $F_x: A \rightarrow B$, not necessarily normalised or sub-normalised.
    A \textbf{controlled process} for the family is a triple $(G, p, m)$ consists of a sharp preparation-observation (SPO) pair $(p: X \rightarrow H, \, m: H \rightarrow X)$---where $X$ is a classical system and $H$ is a generic system---together with a process $G: H \otimes A \rightarrow H \otimes B$ satisfying the following equations:
	\begin{equation} \label{eqn:coherentcontrol1}
		 \input{PhiPicontrolled21.tikz}
         \hspace{5mm} = \hspace{5mm}
         \sum_{x \in X}
         \hspace{5mm}
         \input{PhiPicontrolleddecomposition.tikz}
         \hspace{3cm}
		 \input{PhiPicontrolled2.tikz}
         \hspace{5mm} = \hspace{5mm}
         \sum_{x \in X}
         \hspace{5mm}
         \input{PhiPicontrolleddecomposition2.tikz}
	\end{equation}
    Conversationally, we will also say that $(G, p, m)$ is a ``control of'' the family $(F_x)_{x \in X}$.
\end{definition}

\noindent
The system $H$ acts as a physical control system, while the classical system $X$ contains the logical information about the process choice.
The SPO pair is used to encode the logical information into the physical system and/or to decode it from the physical system.
From the point of view of an actor who is using the SPO pair to encode/decode the logical information, the controlled process is no different than classical control.

\begin{example}\label{example:classically-controlled-process}
    The following \textbf{classically controlled process} always exists in every probabilistic theory:
    \begin{equation}
        \sum_{x \in X}
        \hspace{5mm}
        \input{classicalcontrol.tikz}
    \end{equation}
    Conversationally, we will also refer to the above as \emph{the} ``classical control of'' the family $(F_x)_{x \in X}$.
    Note that if $(G, p, m)$ is any control of $(F_x)_{x \in X}$ then the triple $\left((m \otimes \id{B}) \circ G \circ (p \otimes \id{A}), \id{X}, \id{X}\right)$ is always the classical control of the same $(F_x)_{x \in X}$.
\end{example}

\noindent
The definition allows for much more general notions of control, as we shall shortly see, but it also limits the amount of leakage between the input/output systems $A, B$ and the physical control system $H$.
In particular, an agent without access to the output system $B$ cannot, through the SPO pair alone, extract any information about the input state on system $A$ if the maps $F_x$ are normalised:

\begin{equation}
	\sum_{x \in X}
    \hspace{5mm}
    \input{PhiPicontrolleddecomposition2nokickback.tikz}
    \hspace{5mm} = \hspace{5mm}
    \sum_{x \in X}
    \hspace{5mm}
    \input{identityclassical2a.tikz}
    \hspace{5mm} = \hspace{5mm}
    \input{identityclassical2b.tikz}
\end{equation}

\section{Coherent Control in Quantum Theory}
\label{section:coherent-control-quantum}

In quantum information and computing, the idea of coherently controlling a family of unitary processes (and more generally pure CP maps) is certainly not a new one \cite{abbott2018,oi2003,Araujo2014,Thompson2018}.

\begin{example}\label{example:coherently-controlled-isometry}
    If $(F_x)_{x \in X}$ is a family of pure CP maps in quantum theory (e.g. isometries), a generic \textbf{coherently controlled process} for the family uses $\complexs^X$ as a control system and takes the following form:
    \begin{equation} \label{eqn:controlledisometry}
        \input{controlledisometry.tikz}
    \end{equation}
    where $\alpha$ denotes an arbitrary phase in the canonical basis $(\ket{x})_{x \in X}$ for the control system.
    Conversationally, we will also say that the above is a ``coherent control of'' the family $(F_x)_{x \in X}$.
\end{example}

\newcounter{proposition_coherent_control_pure_cp}
\setcounter{proposition_coherent_control_pure_cp}{\value{theorem_c}}
\begin{proposition}\label{proposition:coherent-control-pure-cp}
    Let $(F_x)_{x \in X}$ be a family of pure CP maps in quantum theory, not necessarily normalised (i.e. not necessarily trace-preserving).
    Assume that $(G,p,m)$ is a controlled process for the family with control system $\complexs^X$, where $(p, m)$ is the SPO for the canonical basis $(\ket{x})_{x \in X}$ and where $G$ is itself a pure CP map.
    Then $G$ takes form~\eqref{eqn:controlledisometry} for some phase $\alpha$.
\end{proposition}

\noindent
Moving away from pure maps, the question arises what the ``coherent'' control for arbitrary CP maps should be.
We adopt the following definition within our framework.

\begin{definition}\label{definition:coherent-control-cp}
    Let $(F_x)_{x \in X}$ be a family of CP maps in quantum theory, not necessarily normalised (i.e. not necessarily trace-preserving).
    A \textbf{coherently controlled process} for the family is a controlled process $(G, p, m)$ for the family such that $G$ is obtained as $\Tr_E(G')$ for some pure CP map $G': H \otimes A \rightarrow H \otimes B \otimes E$.
    Conversationally, we will also say that the above is a ``coherent control of'' the family $(F_x)_{x \in X}$.
\end{definition}

\noindent
If each map $F_x$ in a family $(F_x)_{x \in X}$ of CP maps comes with a chosen purification $F_x = \Tr_E(\hat{F}_x)$---without loss of generality, using the same environment system $E$ for all purifications---then it is easy to construct a coherent control of the family, by taking $(G', p, m)$ to be the coherent control of the family $(\hat{F}_x)_{x \in X}$ of purifications and then discarding the environment $E$ of $G'$.
However, it would be desirable for such a construction to be a function of the family $(F_x)_{x \in X}$ alone, without dependence on additional information.
As shown by Proposition~\ref{proposition:coherent-control-pure-cp}, this is indeed possible when all CP maps in the family are pure.
However, the following no-go result shows this to no longer be the case when generic CP maps are considered.

\newcounter{proposition_coherent_control0cp_nogo}
\setcounter{proposition_coherent_control0cp_nogo}{\value{theorem_c}}
\newcounter{eq_coherent_control0cp_nogo}
\setcounter{eq_coherent_control0cp_nogo}{\value{equation}}
\begin{proposition}\label{proposition:coherent-control-cp-nogo}
    Let $(F_x)_{x \in X}$ be a family of CP maps in quantum theory.
    It is not generally possible to construct a coherent control of the family which is a function of the family $(F_x)_{x \in X}$ alone, i.e. one which is independent of a choice of purification for the CP maps in the family.
    This is the same as the statement that it is not generally possible to construct a coherent control $(\Tr_E(G'), p, m)$ of the family $(F_x)_{x \in X}$ in such a way that the following equation holds for all choices of unitaries $U_x: E \rightarrow E$:
    \begin{equation}\label{eq:coherent-control-cp-nogo}
        \input{coherentcontrolnogo1.tikz}
        \hspace{5mm} = \hspace{5mm}
        \input{coherentcontrolnogo2.tikz}
    \end{equation}
    Note that the unitaries $(U_x)_{x \in X}$ correspond to all possible choices of purification for the CP maps $(F_x)_{x \in X}$.
\end{proposition}

\noindent
The classical control of arbitrary families of CP maps is trivially possible, as shown in Example~\ref{example:classically-controlled-process}.
The formulation of Proposition~\ref{proposition:coherent-control-cp-nogo} shows that, on the other hand, the question of coherently controlling families of CP maps is much more sophisticated, leading to some confusion in the literature about its feasibility.

In \cite{oi2003}, for example, the author interprets the failure to construct such a coherent control independently of the choice of purification (aka choice of Kraus operators) as a sign that an interferometric realisation of such coherent control would extract information about the underlying physical implementation of the CP maps themselves.
We believe that this statement can be easily misinterpreted: the CP maps involved in the experiment are already the ``physical'' ones---defined on the direct sum of the vacuum sector and the 1-particle sector---and the results of the experiment are independent of the choice of purification for them.
This is obvious, since the experiment itself can be easily written as a circuit.

What the results of the experiment actually depend on is the choice of purification for the ``logical'' CP maps involved, those restricted to the 1-particle sector.
This is due to the specific design of the experiment: the implementation of the ``physical'' CP maps is such that they react to the vacuum state on their input by emitting a non-vacuum state $\ket{e}$ on the environment $E$.
This is not physically unreasonable, e.g. if the environment system comprises some static massive particle which is made to interact with the photons passing in the interferometric setup.
However, this dependence on the choice of purification for the ``logical'' CP maps goes away as soon as we allow the environment ``rest'' state $\ket{e}$ to be transformed covariantly with the choice of purification, i.e. if we set $\ket{e} \mapsto U \ket{e}$ whenever we change the purification by applying a unitary $U: E \rightarrow E$ to the environment.

It is this last observation which helps us frame the discussion by \cite{oi2003} within the context of Proposition~\ref{proposition:coherent-control-cp-nogo}: given the two CP maps, acting on the vacuum and non-vacuum sectors respectively, it is very much possible to find a coherent control which is invariant under application of the same unitary $U$ to the environment of both purifications.
This is always the case: if all $U_x: E \rightarrow E$ are chosen to be equal to some fixed $U$, then Equation~\eqref{eq:coherent-control-cp-nogo} always holds (because $U$ is trace-preserving).
What is found to be impossible in the discussion by \cite{oi2003} is to choose such coherent control in a way which is invariant under application of $U$ to the environment of the non-vacuum sector map and of the identity to the environment of the vacuum sector map.
This issue indeed generalises and formed the inspiration for our proof of Proposition~\ref{proposition:coherent-control-cp-nogo}.

\section{Definite and Indefinite Causal Scenarios}

When operational scenarios with definite causal order are depicted diagrammatically in the context of probabilistic theories, it is easy to conflate the boxes in the diagrams with processes happening locally at events (i.e. points in spacetime), and the wires in the diagrams with the information flow establishing the causal relationships between said events.
It has been previously argued \cite{Pinzani2019} that this practice---though natural and notationally pleasant---is not mathematically well-founded, as there need not be a canonical way to decide how a process should be decomposed into a diagram compatibly with a given definite causal structure.
As a consequence, two ingredients are needed when talking about such operational scenarios:

\begin{enumerate}
    \item[(i)] a causal graph, representing the events in the scenario and their definite causal order;
    \item[(ii)] a map assigning each event in the scenario to the process happening at that event.
\end{enumerate}

\noindent
The mathematical structure introduced in \cite{Pinzani2019} as the substrate for such operational scenarios is that of \emph{framed causal graphs}.
For reasons which will become clear later on, we generalise the original definition to include the possibility of multiple edges between the same pairs of event.
Furthermore, we include some additional information about the classical interface of the local processes/experiments, in the form of finite sets of input values that can be used to control them and output values for their outcomes.

\begin{definition}\label{definition:framed-multigraph}
    A \textbf{framed multigraph} is an directed multigraph $\Gamma$
    \footnote{A directed multigraph $\Gamma$ consists of a set $\nodes{\Gamma}$, a set of $\edges{\Gamma}$ and a pair of functions $\tailSym: \edges{\Gamma} \rightarrow \nodes{\Gamma}$ and $\headSym: \edges{\Gamma} \rightarrow \nodes{\Gamma}$ specifying the tail and head of each edge respectively.}
    equipped with the following data:
    \begin{itemize}
    	\item a sub-set $\inputs{\Gamma} \subseteq \nodes{\Gamma}$ of the nodes of $\Gamma$---the \textbf{input nodes}---such that each $x \in \inputs{\Gamma}$ has zero incoming edges and a single outgoing edge;
    	\item a sub-set $\outputs{\Gamma} \subseteq \nodes{\Gamma}$ of the nodes of $\Gamma$---the \textbf{output nodes}---such that each $x \in \outputs{\Gamma}$ has zero outgoing edges and a single incoming edge;
    	\item a \textbf{framing} for $\Gamma$, which consists of the following:
    	\begin{itemize}
    		\item a total order on $\inputs{\Gamma}$;
    		\item a total order on $\outputs{\Gamma}$;
    		\item for each node $x \in \nodes{\Gamma}$, a total order on the edges outgoing from $x$;
    		\item for each node $x \in \nodes{\Gamma}$, a total order on the edges incoming to $x$;
    	\end{itemize}
    \end{itemize}
    We refer to nodes in $\inputs{\Gamma}$ or in $\outputs{\Gamma}$ as \textbf{boundary nodes} and to all other nodes in $\Gamma$ as \textbf{internal nodes}.
    An \textbf{acyclic framed multigraph} is a framed multigraph which is acyclic (and in particular has no loops).
\end{definition}
\begin{remark}
    The input and output nodes of a framed multigraph are designed to behave as ``half-edges'': when two framed multigraphs $\Gamma$ and $\Gamma'$ are composed sequentially, the outputs of $\Gamma$ and the inputs of $\Gamma'$ are joined and disappear, each pair of corresponding output/input resulting in a single edge of the composite framed multigraph $\Gamma' \circ \Gamma$.
    (We do not use such composition here.)
\end{remark}

\begin{definition}\label{definition:definite-causal-scenario}
    A \textbf{definite causal scenario} is a triple $\Theta = (\Gamma, \underline{I}, \underline{O})$ of an acyclic framed multigraph $\Gamma$ with:
    \begin{itemize}
        \item a finite set $I_\omega$ of \textbf{classical inputs} for each $\omega \in \events{\Theta}$, i.e. the values available locally to control the process at the event;
        \item a finite set $O_\omega$ of \textbf{classical outputs} for each $\omega \in \events{\Theta}$, i.e. the values that the process at the event can return locally as its outcome.
    \end{itemize}
    In the above, we have defined the \textbf{events} in the scenario as the set $\events{\Theta} := \nodes{\Gamma} \backslash (\inputs{\Gamma} \sqcup \outputs{\Gamma})$ of internal nodes for $\Gamma$.
\end{definition}

\begin{remark}
    Compared to the original \cite{Pinzani2019}, we have restricted our attention to \emph{chronology respecting} scenarios, i.e. those corresponding to acyclic framed multigraphs.
    However, one could easily extend Definition~\ref{definition:definite-causal-scenario} to one for \emph{chronology violating} scenarios, by allowing the framed multigraph to be cyclic and/or to have loops.
\end{remark}

\noindent
We now define exactly what it means to ``draw a diagram over'' one such definite causal scenario, with semantics valid in any probabilistic theory.

\begin{definition}\label{definition:diagram-over-definite-causal-scenario}
    Let $\Theta = (\Gamma, \underline{I}, \underline{O})$ be a definite causal scenario and let $\mathcal{C}$ be a probabilistic theory.
    A \textbf{diagram over $\Theta$ in $\mathcal{C}$} is a pair of functions $\sysmapSym: \edges{\Gamma} \rightarrow \obj{\mathcal{C}}$ and $\procmapSym: \events{\Gamma} \rightarrow \mor{\mathcal{C}}$, associating each $e \in \edges{\Gamma}$ to a system $\sysmap{e}$ in $\mathcal{C}$ and each $\omega \in \events{\Gamma}$ to a process $\procmap{\omega}$ in $\mathcal{C}$ with the following type:
    \begin{equation}
        \procmap{\omega}:
        I_\omega
        \otimes
        \bigotimes_{e \in \inputs{\omega}} \sysmap{e}
        \longrightarrow
        O_\omega
        \otimes
        \hspace{-2mm}\bigotimes_{e' \in \outputs{\omega}}\hspace{-2mm} \sysmap{e'}
    \end{equation}
    Above we denoted by $\inputs{\omega} := \suchthat{e \in \edges{\Gamma}}{\head{e} = \omega}$ the edges of $\Gamma$ coming into $\omega$ and we similarly denoted by  $\outputs{\omega} := \suchthat{e \in \edges{\Gamma}}{\tail{e} = \omega}$ the edges of $\Gamma$ going out of $\omega$.
\end{definition}

\noindent
Even though it specifies concrete processes in a probabilistic theory, the definition of diagram $\Delta$ above is still partly abstract, as it does not explicitly state how the processes fit together.
What gives it fully concrete semantics is the following definition of the overall process $\diagramproc{\Delta}$ \emph{associated} to the diagram $\Delta$.
See Figure~\ref{fig:definite-diagram} for an exemplification.

\begin{figure}[h]
    \begin{center}
        \input{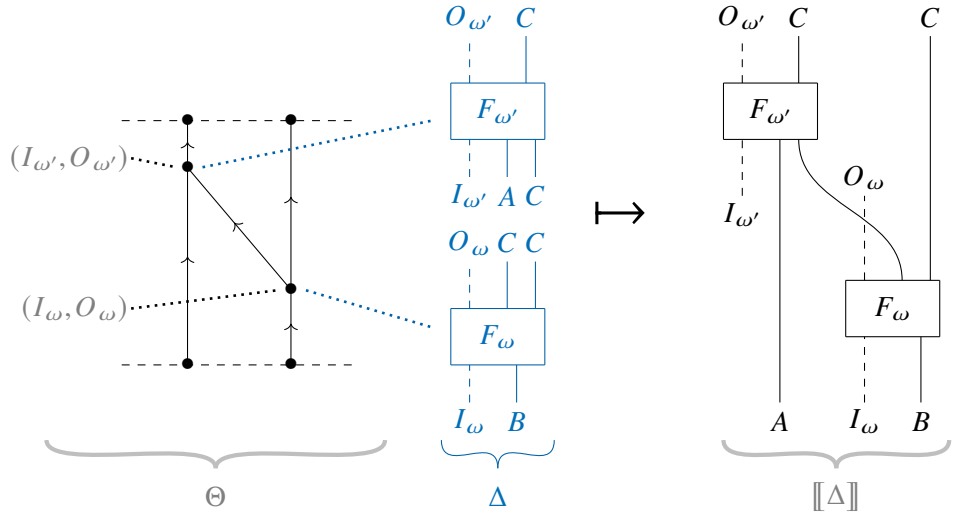}
    \end{center}
    \caption{Graphical exemplification of the association of the process $\diagramproc{\Delta}$ to a diagram $\Delta$ over a definite causal scenario $\Theta$.}
    \label{fig:definite-diagram}
\end{figure}

\begin{definition}\label{definition:process-diagram-over-causal-graph}
    Let $\Theta = (\Gamma, \underline{I}, \underline{O})$ be a definite causal scenario, $\mathcal{C}$ be a probabilistic theory and $\Delta = (\sysmapSym, \procmapSym)$ be a diagram over $\Gamma$ in $\mathcal{C}$.
    The \textbf{process associated to $\Delta$} is the unique process $\diagramproc{\Delta}$ in $\mathcal{C}$ obtained by joining the outputs and inputs of the processes $\procmap{\omega}$ in the diagram $\Delta$ according to the directed multigraph $\Gamma$, resulting in a process with the following overall type:
    \begin{equation}
        \diagramproc{\Delta}:
        \left(
        \left(
        \bigotimes_{\omega \in \events{\Gamma}}
        I_\omega
        \right)
        \otimes
        \left(
        \bigotimes_{e \in \inputs{\Gamma}} \sysmap{e}
        \right)
        \right)
        \longrightarrow
        \left(
        \left(
        \bigotimes_{\omega \in \events{\Gamma}}
        O_\omega
        \right)
        \otimes
        \left(
        \bigotimes_{e' \in \outputs{\Gamma}} \sysmap{e'}
        \right)
        \right)
    \end{equation}
    Each classical input system $I_\omega$ of $\diagramproc{\Delta}$ is wired to the classical input system $I_\omega$ of the process $\procmap{\omega}$.
    Similarly, each classical output system $O_\omega$ of $\diagramproc{\Delta}$ is wired to the classical output system $O_\omega$ of $\procmap{\omega}$.
\end{definition}

Definite causal scenarios and diagrams over them are perfectly adequate when it comes to discussion of operational scenarios over definite causal orders, but they don't have the necessary flexibility to accommodate operational scenarios where causality is indefinite.
The main contribution of this work will now be to define diagrams over \emph{indefinite causal scenarios}, giving them semantics in probabilistic theories using controlled processes.
As a special case, we will be able to describe the idea of superposition of causal orders in quantum theory, i.e. we give solid mathematical foundations over which to discuss the possibility of operational scenarios where local quantum experiments are connected by a superposition of background spacetimes.

Firstly, we define the purely operational canvas against which the indefinite causal scenario takes places.
This includes all the black-box information locally available to the actors in our scenario, but does not include any information about causal order nor any information about the specific implementation of the local processes.

\begin{definition}\label{definition:causal-scenario}
    An \textbf{indefinite causal scenario} $\Phi$ is specified by the following data:
    \begin{itemize}
        \item the set $\Omega$ of \textbf{events} at which the processes (operations, experiments, etc.) take place;
        \item a set $\Xi$ of \textbf{system labels}, used to abstractly indicate which physical systems are guaranteed to be the same across different experiments;
        \item a finite set $I_\omega$ of \textbf{classical inputs} for each $\omega \in \Omega$, i.e. the values available locally to control the process at the event;
        \item a finite set $O_\omega$ of \textbf{classical outputs} for each $\omega \in \Omega$, i.e. the values that the process at the event can return locally as its outcome;
        \item a finite sequence $\Sigma^{in}_\omega$ of elements of $\Xi$ for each $\omega \in \Omega$, the system labels for the physical systems coming into the process from outside the event;
        \item a finite sequence $\Sigma^{out}_\omega$ of elements of $\Xi$ for each $\omega \in \Omega$, the system labels for the physical systems coming out of the process and leaving the event;
        \item a finite sequence $\Pi^{in}$ of elements of $\Xi$, the system labels for the physical systems coming in from outside the region where the scenario is taking place;
        \item a finite sequence $\Pi^{out}$ of elements of $\Xi$, the system labels for the physical systems going out from the region where the scenario is taking place.
    \end{itemize}
    Formally, the scenario is the tuple $\Phi = (\Omega, \Xi, \underline{I}, \underline{O}, \underline{\Sigma}^{in}, \underline{\Sigma}^{out}, \Pi^{in}, \Pi^{out})$, where the underlined letters indicate $\Omega$-indexed families (e.g. $\underline{I} := \omega \mapsto I_\omega$).
\end{definition}

\begin{remark}\label{remark:definite-causal-scenarios-special-case}
    Even though the definitions involved are formally different, the definite causal scenarios of Definition~\ref{definition:definite-causal-scenario} arise naturally as a special case of the indefinite causal scenarios from Definition~\ref{definition:causal-scenario}.
    Indeed, consider an indefinite causal scenario $\Phi = (\Omega, \Xi, \underline{I}, \underline{O}, \underline{\Sigma}^{in}, \underline{\Sigma}^{out}, \Pi^{in}, \Pi^{out})$ and assume that each symbol $e \in \Xi$ appears exactly twice as follows:
    \begin{itemize}
        \item it appears once either in an output set $\Sigma^{out}_{\tail{e}}$ for some $\tail{e} \in \Omega$ or otherwise at some place $\tail{e} \in \{1, ..., |\Pi^{in}|\}$ in the sequence $\Pi^{in}$;
        \item it appears once either in an input set $\Sigma^{in}_{\head{e}}$ for some $\head{e} \in \Omega$ or otherwise at some place $\head{e} \in \{1, ..., |\Pi^{out}|\}$ in the sequence $\Pi^{out}$.
    \end{itemize}
    This defines a framed multigraph $\Gamma$ with $\nodes{\Gamma} := \Omega \sqcup \{1, ..., |\Pi^{in}|\} \sqcup \{1, ..., |\Pi^{out}|\}$ and $\edges{\Gamma} := \Xi$.
    The internal nodes of $\Gamma$ are the events in $\Omega$, yielding a definite causal scenario $(\Gamma, \underline{I}, \underline{O})$. Conversely, each definite causal scenario $(\Gamma, \underline{I}, \underline{O})$ can be turned into an indefinite causal scenario by taking $\Xi := \edges{\Gamma}$ and defining the sequences $\Sigma^{in}_\omega$, $\Sigma^{out}_\omega$, $\Pi^{in}$ and $\Pi^{out}$ from $\inputs{\omega}$, $\outputs{\omega}$, $\inputs{\Gamma}$ and $\outputs{\Gamma}$ respectively.
\end{remark}

\noindent
Secondly, we define the set of definite causal scenarios which are \emph{compatible} with a given indefinite causal scenario: they correspond exactly to all possible ways of joining the output and input physical systems into a multigraph in such a way as to respect the system labels for the indefinite causal scenario.
Each indefinite causal scenario gives rise to a set of compatible definite causal scenarios, each definite causal scenario equipped with a labelling associating each edge of the multigraph to the corresponding system label from the indefinite causal scenario.
See Figure~\ref{fig:definite-assoc} for an exemplification.

\begin{figure}[h]
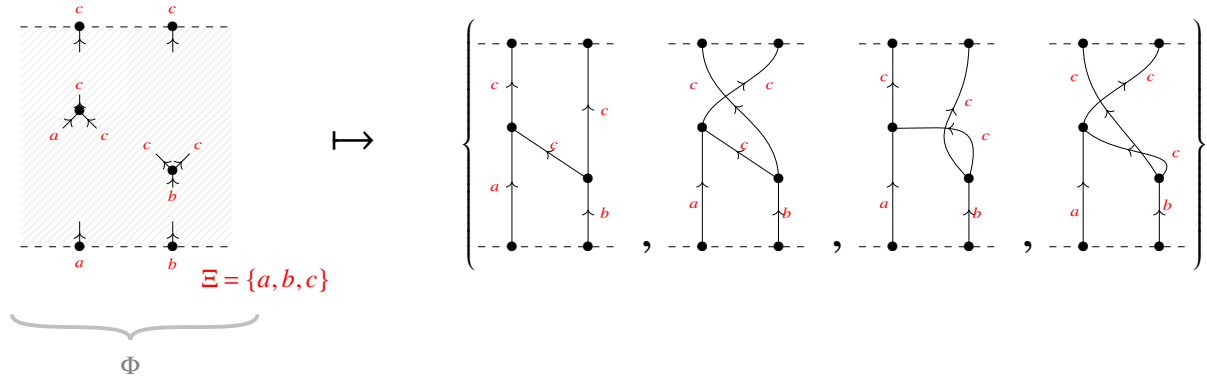

    \begin{center}
    \scalebox{0.9}{$
        \input{Figure2a.tikz}
        \hspace{0mm}
        \mathlarger{\mathlarger{\mathlarger{\mathlarger{\mapsto}}}}
        \hspace{12mm}
        \left\{ \input{graph1.tikz}
        \hspace{2mm}
        \input{graph2.tikz}
        \hspace{2mm}
        \input{graph3.tikz}
        \hspace{2mm}
        \input{graph4.tikz}
        \right\}
    $}
    \end{center}
    \caption{Graphical exemplification of the association of definite causal scenarios to an indefinite causal scenario $\Phi$.}
    \label{fig:definite-assoc}
\end{figure}

\begin{definition}\label{definition:definite-compatibl}
    Let  $\Phi = (\Omega, \Xi, \underline{I}, \underline{O}, \underline{\Sigma}^{in}, \underline{\Sigma}^{out}, \Pi^{in}, \Pi^{out})$ be an indefinite causal scenario.
    A definite causal scenario $\Theta = (\Gamma, \underline{I}', \underline{O}')$ is \textbf{compatible} with $\Phi$ if the following conditions hold:
    \begin{enumerate}
        \item[(i)] the events of $\Theta$ (internal nodes of $\Gamma$) are the events of the scenario, i.e. we have $\events{\Theta} = \Omega$;
        \item[(ii)] we have $\underline{I}' = \underline{I}$ and $\underline{O}' = \underline{O}$;
        \item[(iii)] for each $\omega \in \Omega$, $\inputs{\omega}$ and $\Sigma^{in}_\omega$ have the same number of elements; write $\Sigma^{in}_\omega(e)$ for the system label in the totally ordered set $\Sigma^{in}_\omega$ at the same position as edge $e$ in the totally ordered set $\inputs{\omega}$;
        \item[(iv)] for each $\omega \in \Omega$, $\outputs{\omega}$ and $\Sigma^{out}_\omega$ have the same number of elements; write $\Sigma^{out}_\omega(e)$ for the system label in the totally ordered set $\Sigma^{out}_\omega$ at the same position as edge $e$ in the totally ordered set $\outputs{\omega}$;
        \item[(v)] the input nodes $\inputs{\Gamma}$ and $\Pi^{in}$ have the same number of elements; write $\Pi^{in}(e)$ for the system label in the totally ordered set $\Pi^{in}$ at the same position as $\tail{e}$ in the totally ordered set $\inputs{\Gamma}$;
        \item[(vi)] the output nodes $\outputs{\Gamma}$ and $\Pi^{out}$ have the same number of elements; write $\Pi^{out}(e)$ for the system label in the totally ordered set $\Pi^{out}$ at the same position as $\head{e}$ in the totally ordered set $\outputs{\Gamma}$;
        \item[(vii)] for each edge $e \in \edges{\Gamma}$, the system label at its tail and at its head coincide, i.e. we have $\Sigma^{out}_{\tail{e}}(e) = \Sigma^{in}_{\head{e}}(e)$; by convention, we set $\Sigma^{out}_{\tail{e}}(e) := \Pi^{in}(e)$ when $\tail{e}$ is an input node and $\Sigma^{in}_{\head{e}}(e) := \Pi^{out}(e)$ when $\head{e}$ is an output node (both can be true at the same time).
    \end{enumerate}
    A definite causal scenario $\Theta$ which is compatible with the indefinite causal scenario $\Phi$ comes equipped with an edge labelling $\syslabelSym{\Phi}{\Theta}: \edges{\Gamma} \rightarrow \Xi$, sending each edge $e \in \edges{\Gamma}$ to the system label $\syslabel{\Phi}{\Theta}{e} \in \Xi$ which the indefinite causal scenario associates to the endpoints of the edge.
    We write $\definiteCausAssoc{\Phi}$ for the set of definite causal scenarios compatible with $\Phi$.
\end{definition}

Finally, we are in a position to define what it means to draw a diagram over an \emph{indefinite} causal scenario $\Phi$.
This is a generalisation of the abstract notion of drawing a diagram over a definite causal scenario from Definition~\ref{definition:diagram-over-definite-causal-scenario}: processes are associated to the events, but now taking additional care that the induced diagrams over all \emph{definite} scenarios compatible with $\Phi$ are well-defined.
See Figure~\ref{fig:indefinite-diagram} for an exemplification.

\begin{figure}[h]
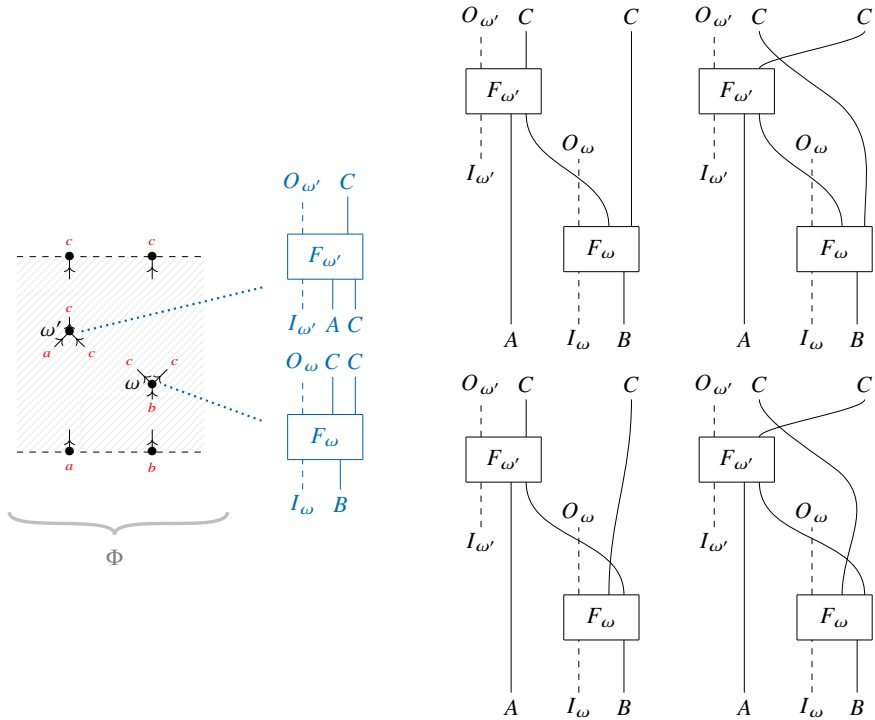

    \begin{center}
    \scalebox{0.8}{$
        \input{Figure3a.tikz}
        \hspace{15mm}
        \input{Figure3b.tikz}
    $}
    \end{center}
    \caption{Graphical exemplification of how a diagram over an indefinite causal scenario $\Phi$ gives rise to diagrams over the compatible definite causal scenarios (cf. Figure~\ref{fig:definite-assoc}).}
    \label{fig:indefinite-diagram}
\end{figure}

\begin{definition}\label{definition:diagram-over-indefinite-causal-scenario}
    Let $\Phi = (\Omega, \Xi, \underline{I}, \underline{O}, \underline{\Sigma}^{in}, \underline{\Sigma}^{out}, \Pi^{in}, \Pi^{out})$ be an indefinite causal scenario and let $\mathcal{C}$ be a probabilistic theory.
    A \textbf{diagram over $\Phi$ in $\mathcal{C}$} is a pair of functions $\sysmapSym: \Xi \rightarrow \obj{\mathcal{C}}$ and $\procmapSym: \Omega \rightarrow \mor{\mathcal{C}}$, associating each system label $\xi$ to a system $\sysmap{\xi}$ in $\mathcal{C}$ and each event $\omega \in \Omega$ to a process $\procmap{\omega}$ in $\mathcal{C}$ with the following type:
    \begin{equation}
        \procmap{\omega}:
        I_\omega
        \otimes
        \bigotimes_{\xi \in \Sigma^{in}_{\omega}} \sysmap{\xi}
        \longrightarrow
        O_\omega
        \otimes
        \hspace{-2mm} \bigotimes_{\xi' \in \Sigma^{out}_{\omega}} \hspace{-2mm}\sysmap{\xi'}
    \end{equation}
    If $\Delta$ is such a diagram and $\Theta = (\Gamma,\underline{I}, \underline{O})$ is a definite causal scenario compatible with $\Phi$, then the \textbf{induced diagram} $\restrict{\Delta}{\Theta} = (\restrict{\sysmapSym}{\Theta}, \restrict{\procmapSym}{\Theta})$ over $\Theta$ is defined by setting $\restrict{\sysmapSym}{\Theta}: \edges{\Gamma} \rightarrow \obj{\mathcal{C}}$ to be $\restrict{\sysmapSym}{\Theta}\left(e\right) := \sysmap{\syslabel{\Phi}{\Theta}{e}}$ and setting $\restrict{\procmapSym}{\Theta}: \events{\Theta} \rightarrow \mor{\mathcal{C}}$ to be $\restrict{\procmapSym}{\Theta}\left(\omega\right) := \procmap{\omega}$.
\end{definition}

\noindent
The semantics for a diagram over an indefinite causal scenario $\Phi$ will no longer be given by a single process with a definite causal order---as was the case for the semantics of diagrams over definite scenarios---but rather a controlled process for the family of all diagrams over all definite scenarios compatible with $\Phi$.

\begin{definition}\label{definition:controlled-process-diagram-over-indefinite-scenario}
    Let $\Phi$ be an indefinite causal scenario, $\mathcal{C}$ be a probabilistic theory and $\Delta = (\sysmapSym, \procmapSym)$ be a diagram over $\Phi$ in $\mathcal{C}$.
    A \textbf{controlled process associated to $\Delta$} is a controlled process $(G, p, m)$ in $\mathcal{C}$ associated to the family $\left(\diagramproc{\restrict{\Delta}{\Theta}}\right)_{\Theta \in \definiteCausAssoc{\Phi}}$ of induced diagrams over all definite causal scenarios $\Theta$ compatible with $\Phi$.
\end{definition}

It is worth noting that the semantics for diagrams over indefinite causal scenarios result in a \emph{controlled} process $(G, p, m)$: the control system is left open on the side, allowing preparations and observations to be used to control the causal order in all possible ways.
Regardless of the specific controlled process and regardless of the specific probabilistic theory, the following will always work.

If we pre-compose the controlled process $G$ with $p$ on the control system we are able to classically control the causal order.
Specifically, feeding a specific value $\Theta \in \definiteCausAssoc{\Phi}$ as input to $p$ results in the diagram $\restrict{\Delta}{\Theta}$ for the definite causal scenario $\Theta$.
More generally, feeding a a probability distribution $\sum_{\Theta \in \definiteCausAssoc{\Phi}} \mathbb{P}(\Theta) \delta_\Theta$ as input to $p$ results in a convex mixture of the diagrams associated to the definite scenarios, each diagram $\restrict{\Delta}{\Theta}$ happening with probability $\mathbb{P}(\Theta)$.

If we pre-compose the controlled process $G$ with a normalised state $\rho$ on the control system and we post-compose it with the observation $m$ on the control system, we obtain again a convex mixture of the diagrams associated with the definite scenarios, with probability distribution given by the classical state $m \circ \rho$.
If instead of the observation $m$ we use the discarding map on the control system, we obtain the same convex mixture, but now without being able of extracting information about the causal order from the classical outcome of the observation $m$.
\footnote{This is because applying the discarding map on the control system is the same as first applying the observation $m$ and then discarding its classical outcome, resulting in a mixture.}

The last observation, that discarding the control system always yields a convex mixture of causal orders, will play an important role in the next Section, when we construct superpositions of causal orders in quantum theory.
Indeed, the observation implies that we cannot obtain superposition of causal orders by preparing the control system in a superposition and then discarding it after we are done.
This is because the act of discarding is an \emph{epistemic} one: it simply means that information about the system is not locally available, not that it can never be recovered.
In order to obtain a \emph{true} superposition, we will have to permanently \emph{erase} the information about definite causal orders, by measuring in an unbiased basis (e.g. the Fourier one).

\section{Superposition of Causal Orders}

In this Section, we show how quantum superposition of causal orders can be modelled within our framework.
We start by looking at the \emph{quantum switch} and then proceed to generalise our construction to arbitrary indefinite causal scenarios.
We conclude by showing that superpositions of causal orders can be constructed purely as a function of the quantum instruments operated at the events, with no dependence on a choice of purification for the CP maps.
This is somewhat surprising---in light of Proposition~\ref{proposition:coherent-control-cp-nogo}---and it is a consequence of the fact that each discarded environment refers to the same local CP map in all branches of the superposition (something which is not true in general coherent control of CP maps, e.g. in the circumstances considered by \cite{oi2003}).

The \emph{quantum switch} is an indefinite causal scenario widely studied in recent literature \cite{Chiribella2013,salek2018,Oreshkov2019}.
In its simplest form it involves two parties---call them Alice and Bob---each operating some quantum instrument locally to their laboratory.
Some physical system comes into the scenario from outside, enters the first instrument (Alice's or Bob's), comes out (possibly altered), enters the second instrument (Bob's or Alice's, respectively), comes out (possibly altered again) and finally leaves the scenario.
The scenario involves the superposition of two definite causal orders: Alice-before-Bob and Bob-before-Alice.
In the $n$-partite generalisation of the quantum switch, $n$ parties operate their quantum instruments sequentially on the same physical system, resulting in a superposition of $n!$ causal orders (corresponding to all possible permutations of the parties).
It is now straightforward to model such a scenario within our framework.

\begin{definition}\label{definition:switch}
    An \textbf{$n$-partite switch} is an indefinite causal scenario satisfying the following conditions:
    \begin{itemize}
        \item there are $n$ events, corresponding to the $n$ parties;
        \item the set $\Xi$ has a single element, as the same physical system is operated upon by all parties;
        \item the sequences $\Sigma^{in}_\omega$, $\Sigma^{out}_\omega$, $\Pi^{in}$ and $\Pi^{out}$ each have a single element, forcing the parties to operate on the same physical system one after the other.
    \end{itemize}
    The classical inputs $I_\omega$ and classical outputs $O_\omega$ are free to choose.
\end{definition}

\noindent
If $\Phi$ is an $n$-partite switch, the semantics of a diagram $\Delta$ over $\Phi$ in a probabilistic theory is given by a controlled process where each choice of permutation $\sigma \in S_n$ for the parties $\{1,...,n\}$ results in party $\sigma(1)$ acting first, followed by party $\sigma(2)$, followed by all other parties in order until party $\sigma(n)$.

For the sake of simplicity, we will now restrict our attention to the bipartite ($n=2$) case.
The two parties are our beloved Alice and Bob, the corresponding events are called $\alpha$ and $\beta$, the physical system will be some quantum system $Z$ and the quantum instruments operated by Alice and Bob will be $F_\alpha: I_\alpha \otimes Z \rightarrow O_\alpha \otimes Z$ and $F_\beta: I_\beta \otimes Z \rightarrow O_\beta \otimes Z$ respectively.
Each quantum instrument $F_\omega: I_\omega \otimes Z \rightarrow O_\omega \otimes Z$ (for $\omega \in \{\alpha, \beta\}$) is defined by a family of CP maps $F_\omega(o_\omega \vert i_\omega): Z \rightarrow Z$ indexed by each possible classical input value $i_\omega \in I_\omega$ and classical output value $o_\omega \in O_\omega$, subject to the normalisation requirement that $\sum_{o_\omega \in O_\omega} F_\omega(o_\omega \vert i_\omega)$ be a CPTP map for each choice of classical input $i_\omega \in I_\omega$.
Let $\hat{F}_\omega(o_\omega \vert i_\omega): Z \rightarrow Z \otimes E_\omega$ be a family of purifications for the CP maps, chosen (without loss of generality) to all have the same environment $E_\omega$.
This results in the following scenario $\Phi$ and diagram $\Delta$:

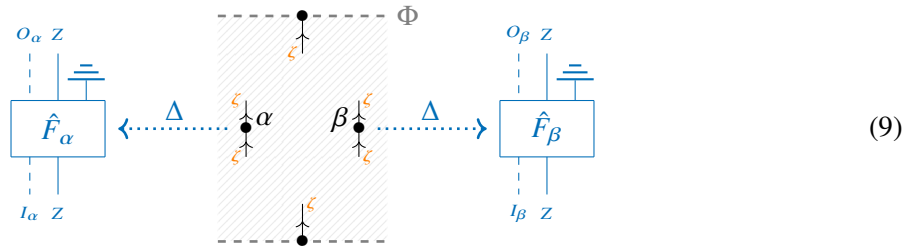
\begin{equation}\label{eq:bipartiteSwitch}
    \input{figure9.tikz}
\end{equation}

\noindent
We can construct a controlled process associated to $\Delta$ by using our definition of coherently controlled processes from Section~\ref{section:coherent-control-quantum}.
Specifically, for each fixed $\underline{i} = (i_\alpha, i_\beta) \in \underline{I}$ and $\underline{o} = (o_\alpha, o_\beta) \in \underline{O}$ we can define the following pure controlled processes from the purifications:

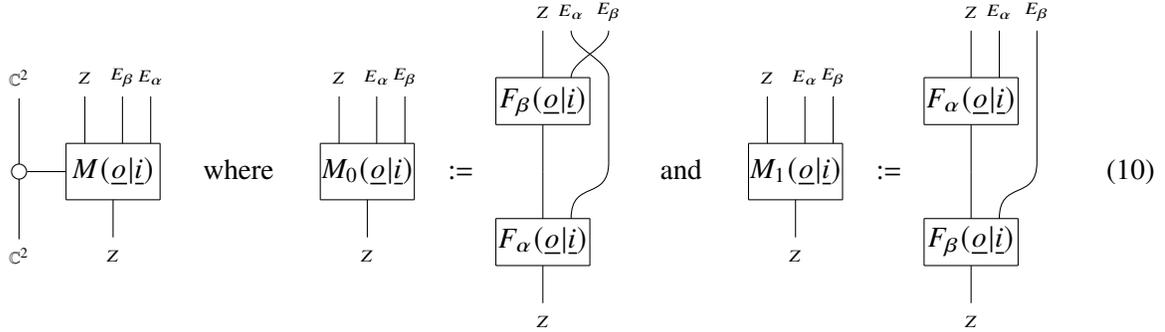
\begin{equation}\label{eq:bipartiteSwitchPureControl}
    \input{figure10a.tikz}
    \hspace{4mm} \textrm{where} \hspace{6mm}
    \input{figure10b.tikz}
    \hspace{2mm} := \hspace{2mm}
    \input{figure10c.tikz}
    \hspace{4mm} \textrm{and} \hspace{6mm}
    \input{figure10d.tikz}
    \hspace{2mm} := \hspace{2mm}
    \input{figure10e.tikz}
\end{equation}

\noindent
If we discard the two environment $E_\alpha$ and $E_\beta$ and we reintroduce the classical inputs and outputs, we obtain the controlled process for $\Delta$:

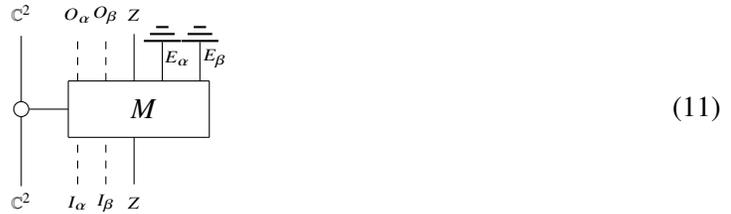
\begin{equation}\label{eq:bipartiteSwitchControl}
    \input{figure11.tikz}
\end{equation}

\noindent
As mentioned in the previous Section, the controlled process above is very general: amongst other things, we can plug any qubit state into the control systems and perform any measurement on it afterwards.
In order to obtain a true superposition of the two causal orders, we use some phase state $\ket{\varphi} := \frac{1}{\sqrt{2}}(\ket{0} + e^{i\varphi} \ket{1})$ and measure in the Pauli X basis.
The two measurement outcomes $\bra{\pm}$ then correspond to families of processes---indexed by the classical inputs $\underline{i} \in \underline{I}$ and classical outputs $\underline{o} \in \underline{O}$---which see a superposition with phase of the two causal orders.
In traditional notation, the processes can be written as follows:

\begin{equation*}
    \scalebox{0.95}{$
    \frac{1}{4} \Tr_{E_\alpha \otimes E_\beta}\Big[
        \operatorname{dbl}\left[\left(\hat{F}_\alpha(o_\alpha | i_\alpha) \otimes \id{E_\beta}\right) \circ \hat{F}_\beta(o_\beta | i_\beta)\right]
        \pm
        e^{i\varphi}
        \operatorname{dbl}\left[\left(\id{Z} \otimes \sigma_{E_\beta, E_\alpha}\right) \circ \left(\hat{F}_\beta(o_\beta | i_\beta) \otimes \id{E_\alpha}\right) \circ \hat{F}_\alpha(o_\alpha | i_\alpha)\right] \Big]
    $}
\end{equation*}

\noindent
where we introduced the short-hand $\operatorname{dbl}[U]$ for the CP map $\operatorname{dbl}[U] := U U^\dagger$ corresponding to a linear map $U$ and we have freely confused the pure CP maps $\hat{F}_\omega(o_\omega | i_\omega)$ with the corresponding linear maps.

Proposition~\ref{proposition:coherent-control-cp-nogo} tells us that, in general, a controlled process such as~\eqref{eq:bipartiteSwitchControl} will depend on our choices of purification $\hat{F}_\omega$.
However, this turns out not to be the case for the switch.
Indeed, we can pull the two environments to the boundary of the scenario, keeping them throughout the superposition of causal orders and only discarding afterwards:

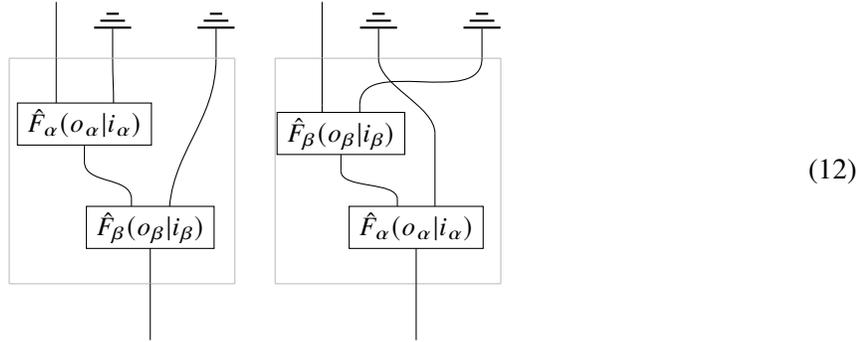
\begin{equation}\label{eq:switchPurification1}
    \input{switchBA.tikz}
    \hspace{5mm}
    \input{switchABnew.tikz}
\end{equation}

\noindent
Any alternative choice of purification can be obtained by applying a suitable unitary to each environment.
However, the same map $\hat{F}_{\omega}$ appears at the bottom of the environment $E_\omega$ in all branches of the superposition, and hence the same unitary appears applied to the environment.
This means that we can pull the unitaries themselves out of the scenario:

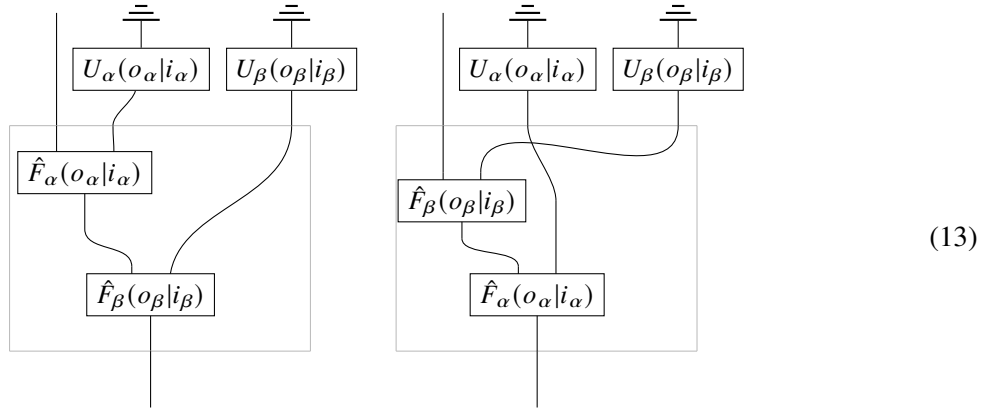
\begin{equation}\label{eq:switchPurification2}
    \input{switchBAnewpurification.tikz}
    \hspace{5mm}
    \input{switchABnewpurification.tikz}
\end{equation}

\noindent
The unitaries will then be cancelled by the discarding maps, leading us to conclude that the controlled process we constructed for the switch was actually independent of our choices of purification $\hat{F}_\omega$ and is therefore a function of the original quantum instruments $F_\omega$.
This argument is not unique to the switch, but instead generalises to coherent control for all diagrams over indefinite causal scenarios in quantum theory, as dictated by our final result below.
As a consequence, our framework can be used to give well-defined semantics to superposition of causal orders in quantum theory.

\begin{definition}\label{definition:purification-scenarios-diagrams}
    Let $\Phi = (\Omega, \Xi, \underline{I}, \underline{O}, \underline{\Sigma}^{in}, \underline{\Sigma}^{out}, \Pi^{in}, \Pi^{out})$ be an indefinite causal scenario and $\Delta$ be a diagram over $\Phi$ in quantum theory.
    The \textbf{purification} of $\Phi$ is the indefinite causal scenario obtained from $\Phi$ by adding fresh symbols $\varepsilon_\omega$ to $\Xi$ and $\Pi^{out}$ for all $\omega \in \Omega$ (in some chosen order).
    A \textbf{purification of $\Delta$ with environments $(E_\omega)_{\omega \in \Omega}$} is the diagram over the purification of $\Phi$ obtained by considering purifications of all CP maps in $\Delta$, with fixed environment $E_\omega$ for each event $\omega \in \Omega$.
    \footnote{The purifications considered here are for each fixed choice of classical input and output values, i.e. they are indexed families of purifications $\hat{F}_\omega(o_\omega | i_\omega)$ for the CP maps $F_\omega(o_\omega | i_\omega)$ corresponding to each given choice of $i_\omega \in I_\omega$ and $o_\omega \in O_\omega$.}
\end{definition}

\noindent
The purification of a diagram---and the necessary corresponding ``purification'' of the underlying indefinite causal scenario---generalise the process seen in~\eqref{eq:switchPurification1} and~\eqref{eq:switchPurification2} above, where the environment wires were ``pulled to the boundary'' of the diagram.
We use purifications of a diagram $\Delta$ to define its coherent control.

\begin{definition}\label{definition:coherent-control-causal-orders}
    Let $\Phi$ be an indefinite causal scenario and $\Delta$ be a diagram over $\Phi$ in quantum theory.
    Let $P(\varphi) := \sum_{\Theta \in \definiteCausAssoc{\Phi}} e^{i \varphi_\Theta}\ket{\Theta}\bra{\Theta}$ be a phase gate for the computational basis of $\complexs^{\definiteCausAssoc{\Phi}}$.
    The \textbf{coherent control of $\Delta$ with phase $P(\varphi)$} is defined to be the coherent control---with phase $P(\varphi)$ and control system $\complexs^{\definiteCausAssoc{\Phi}}$---of the following family of processes, where $\Delta^{pure}$ is a purification of $\Delta$ with environments $(E_\omega)_{\omega \in \Omega}$ and we have defined the global environment $E := \bigotimes\limits_{\omega \in \Omega} E_\omega$:
    \begin{equation}
        \Big(
        \Tr_{E}
        \diagramproc{\restrict{\Delta^{pure}}{\Theta}}
        \Big)_{\Theta \in \definiteCausAssoc{\Phi}}
    \end{equation}
    Note that $\Tr_{E}\diagramproc{\restrict{\Delta^{pure}}{\Theta}}$ is a diagram over $\Theta$ for all $\Theta \in \definiteCausAssoc{\Phi}$.
\end{definition}

\newcounter{proposition_coherent_control0causal_orders}
\setcounter{proposition_coherent_control0causal_orders}{\value{theorem_c}}
\begin{proposition}\label{proposition:coherent-control-causal-orders}
    Let $\Phi$ be an indefinite causal scenario and $\Delta$ be a diagram over $\Phi$ in quantum theory.
    For each possible choice of phase gate $P(\varphi)$ for the computational basis of $\complexs^{\definiteCausAssoc{\Phi}}$, the coherent control of $\Delta$ with phase $P(\varphi)$ is well-defined independently of the choice of purification for the processes in $\Delta$.
\end{proposition}

\noindent
Having shown that the coherent control with phase for diagrams over indefinite causal scenarios is a well-defined concept, we conclude by providing the following definition for sake of clarity and future use.

\begin{definition}\label{definition:superposition-causal-order}
    Let $\Phi$ be an indefinite causal scenario and $\Delta$ be a diagram over $\Phi$ in quantum theory.
    By a \textbf{superposition of causal orders for $\Delta$ with phase $P(\varphi)$} we mean the quantum instrument obtained by:
    \begin{enumerate}
        \item[(i)] considering the coherent control of $\Delta$ with phase $P(\varphi)$;
        \item[(ii)] pre-composing the control system with the uniform superposition state $\ket{+}$ below:
        \begin{equation}
            \ket{+}
            :=
            \frac{1}{\sqrt{|\definiteCausAssoc{\Phi}|}}
            \sum_{\Theta \in \definiteCausAssoc{\Phi}} \ket{\Theta}
        \end{equation}
        \item[(iii)] post-composing the control system with some measurement which is unbiased with respect to the computational basis (e.g. one in the Fourier basis for some group structure on $\definiteCausAssoc{\Phi}$).
    \end{enumerate}
    Note that the choice of fixed input state $\ket{+}$ was made to avoid redundancy in the relative phase between the different causal orders, which is already controlled by $P(\varphi)$.
\end{definition}

\section*{Acknowledgements}

NP would like to acknowledge financial support from EPSRC and the Pirie-Reid Scholarship.
SG would like to acknowledge support from Hashberg Ltd.

\newpage
\bibliography{paperQPL2020}
\bibliographystyle{eptcs}

\appendix

\section{Categorical Probabilistic Theories}
\label{appendix:cpt}

Quantum theory can be formulated using different mathematical formalisms, each highlighting different conceptual aspects of the theory.
In the purely categorical formulation of process theories \cite{Coecke2016}, classical systems emerge as a compositional part of the theory itself, but the probabilistic structure is not explicitly axiomatised.
In the formulation of \emph{operational probabilistic theories} (OPT) by \cite{chiribella2010,chiribella2017}, on the other hand, the probabilistic structure is explicitly axiomatised, but classical systems are \emph{not} treated compositionally as part of the theory, with the quantum-classical interface relegated to a mere labelling of quantum processes.
The formulation of probabilistic theories by \cite{Gogioso2018} tries to take the best of both worlds: the compositional treatment of classical system from categorical formulations and the explicit axiomatisation of probabilistic structure from OPTs.
This leads to the following definition of a probabilistic theory.

\begin{definition}

    A \textbf{probabilistic theory} is a (strict) symmetric monoidal category (SMC) $\CategoryC$ which satisfies the following requirements:

    \begin{itemize}

    	\item there is a full sub-SMC of $\CategoryC$, denoted by $\CategoryC_K$, which is equivalent to the SMC $\RMatCategory{\reals^+}$ modelling \textbf{classical theory} (itself a probabilistic theory);

    	\item the SMC $\CategoryC$ is enriched in commutative monoids and the induced enrichment on $\CategoryC_K$ coincides with the one given by the linear structure of $\RMatCategory{\reals^+}$;

    	\item the SMC $\CategoryC$ comes with an environment structure, i.e with a family of effects $\top_A : A \rightarrow 1$ which satisfy the following requirements:
    	\begin{equation}
    		\input{environmentstructure.tikz}
    	\end{equation}
        The environment structure induced on $\CategoryC_K$ coincides with the one given by marginalisation in $\RMatCategory{\reals^+}$.

    \end{itemize}

    \noindent
    The requirements above imply that processes $A \rightarrow B$ in $\CategoryC$ have the structure of a \emph{convex cone}, i.e. they are $\reals^+$-modules.
    The effects $\top_A : A \rightarrow 1$ are known as \textbf{discarding maps}.
    The systems in the full sub-category $\CategoryC_K$ are known as \textbf{classical systems} and the processes between them as \textbf{classical processes}.
\end{definition}

\noindent
The probabilistic theory most relevant to this work is \emph{quantum theory}, defined by taking CP maps together with classical theory and introducing the quantum-classical interface by linearity from families of quantum processes (cf. resolution of the classical identity below).

From a diagrammatic perspective, dashed wires are used to denote systems which are guaranteed to be classical, while solid wires denote generic systems.
It is convenient to assume the following general form for processes in probabilistic theories, with distinguished classical input and output systems (finite sets $X$ and $Y$ respectively):

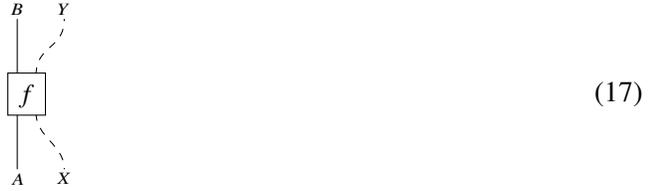
\begin{equation}\label{eq:generalprocessCPT}
    \input{generalprocessCPT.tikz}
\end{equation}

\noindent
The linear structure of $\RMatCategory{\reals^+}$ can be used to explicitly perform resolutions of the classical identity:

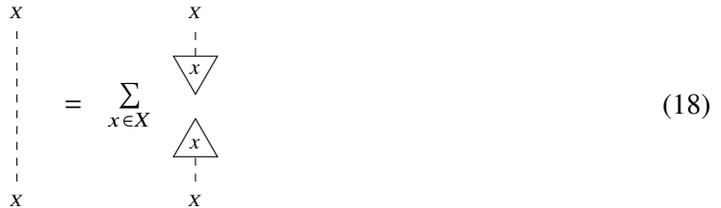
\begin{equation}
	\input{identityclassical.tikz}
\end{equation}

\noindent
The resolution of the classical identity can be used to freely confuse between the following two perspectives, establishing a direct link to the OPT formalism:

\begin{itemize}
    \item processes $F: A \otimes X \rightarrow B \otimes Y$, with $X$ and $Y$ classical systems (i.e finite sets);
    \item families $F(y | x): A \rightarrow B$ of processes indexed by the classical values $x \in X$ and $y \in Y$.
\end{itemize}

\noindent
The discarding maps axiomatise the notion of marginalisation (aka \emph{partial trace} in the context of quantum systems).
They can also be used to define a sub-SMC of \emph{normalised} processes:

\begin{definition}
    A process $f: A \rightarrow B$ is said to be \textbf{normalised} if it satisfies the following equation:
    \begin{equation}
    	\input{normalisedprocess.tikz}
    \end{equation}
    A process $f : A \rightarrow B$ is said to be \textbf{sub-normalised} if there exists some $g: A \rightarrow B$ such that $f+g$ is normalised (in which case $g$ is also sub-normalised).
\end{definition}

\noindent
In particular, the normalised states on a classical system $X$ are the probability distributions on $X$, the normalised processes $X \rightarrow Y$ are the $Y$-by-$X$ stochastic matrices, and discarding on a classical output of a classical process is the same as marginalisation.

The axiomatisation of causality as ``no-signalling from the future'' \cite{chiribella2017} is embodied in probabilistic theories by the following observation about normalised processes:

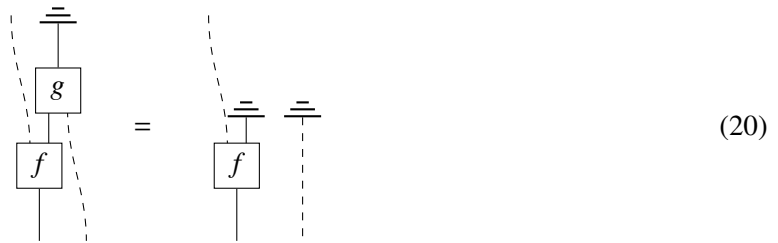
\begin{equation}
	\input{no-signalling-future-CPT.tikz}
\end{equation}

\noindent
In the above, we see that the classical outcome of a test $f$ cannot be influenced by a controlled process $g$ in its future (i.e. it is independent of the classical input used to control $g$).

The framework also adopts a notion of \emph{purity}, defined as the lack of non-trivial interaction with a discarded environment.

\begin{definition}\label{def:purity}
    A process $g: A \rightarrow B$ is said to be \textbf{pure} if whenever we can find a system $E$ and a process $f: A \rightarrow E \otimes B$ such that the following equality holds:
    \begin{equation*}
        \input{purity0.tikz}
    \end{equation*}
    then there exists a normalised state $\psi : 1 \rightarrow E$, dependent on $f$, such that the following equality also holds:
    \begin{equation*}
        \input{purity.tikz}
    \end{equation*}
    Note that neither $g$ nor $f$ are required to be normalised as part of this definition.
\end{definition}

Finally, we recall the definition of a sharp preparation-observation pair, capturing the idea of perfect encoding/decoding of classical information in arbitrary systems.

\begin{definition}\label{definition:spo-pair}
    A \textbf{sharp preparation-observation pair} (SPO pair) is a pair $(p, m)$ of a \textbf{preparation} process $p: X \rightarrow H$ and an \textbf{observation} process $m: H \rightarrow X$ on some classical system $X$ and some arbitrary system $H$, such that the following equation holds:
    \begin{equation}
        \input{SPOa.tikz} 
        \hspace{5mm} = 
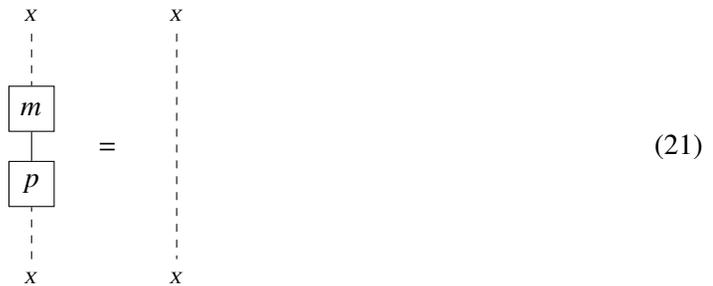
\hspace{5mm}
        \begin{tikzpicture}
	\begin{pgfonlayer}{nodelayer}
		\node [style=none] (1) at (0, -3) {};
		\node [style=none] (2) at (0, 3) {};
		\node [style=eslabel] (3) at (0, -3.5) {$X$};
		\node [style=eslabel] (4) at (0, 3.5) {$X$};
	\end{pgfonlayer}
	\begin{pgfonlayer}{edgelayer}
		\draw [style=classical] (2.center) to (1.center);
	\end{pgfonlayer}
\end{tikzpicture}

    \end{equation}
\end{definition}

\section{Proofs}

\subsection*{Proof of Proposition~\ref{proposition:coherent-control-pure-cp}}

\setcounter{theorem_c}{\value{proposition_coherent_control_pure_cp}}
\begin{proposition}
    Let $(F_x)_{x \in X}$ be a family of pure CP maps in quantum theory, not necessarily normalised (i.e. not necessarily trace-preserving).
    Assume that $(G,p,m)$ is a controlled process for the family with control system $\complexs^X$, where $(p, m)$ is the SPO for the canonical basis $(\ket{x})_{x \in X}$ and where $G$ is itself a pure CP map.
    Then $G$ takes form~\eqref{eqn:controlledisometry} for some phase $\alpha$.
\end{proposition}
\begin{proof}

    In quantum theory, the normalised SPO pairs connecting a classical system $X$ and the quantum system $\complexs^X$ are exactly the preparations and measurements in some orthonormal basis of $\complexs^X$.
    Without loss of generality, assume that the SPO pair is for the canonical basis $(\ket{x})_{x \in X}$ of system $\complexs^X$:

    \begin{equation}\label{eqn:measureprepare}
    	p = \begin{tikzpicture}
	\begin{pgfonlayer}{nodelayer}
		\node [style=none] (0) at (0, 1) {};
		\node [style=none] (1) at (0, -1) {};
		\node [style=white dot] (2) at (0, 0) {};
	\end{pgfonlayer}
	\begin{pgfonlayer}{edgelayer}
		\draw (0.center) to (2);
		\draw [style=classical] (2) to (1.center);
	\end{pgfonlayer}
\end{tikzpicture}
 \hspace{2cm} m = \begin{tikzpicture}
	\begin{pgfonlayer}{nodelayer}
		\node [style=none] (0) at (0, -1) {};
		\node [style=none] (1) at (0, 1) {};
		\node [style=white dot] (2) at (0, 0) {};
	\end{pgfonlayer}
	\begin{pgfonlayer}{edgelayer}
		\draw (0.center) to (2);
		\draw [style=classical] (2) to (1.center);
	\end{pgfonlayer}
\end{tikzpicture}

    \end{equation}

    \noindent
    Now assume that $(G, p, m)$ is a controlled process for the family $(F_x)_{x \in X}$.
    In particular, it satisfies the following equalities for all $x \in X$, where on the right we have used the coherent control of Example~\ref{example:coherently-controlled-isometry}:

    \begin{equation}
    	\input{Gmeasure.tikz}
        \hspace{5mm} = \hspace{5mm}
        \input{controlledisometrydecomposition.tikz}
        \hspace{5mm} = \hspace{5mm}
        \input{controlledisometry2.tikz}
    \end{equation}

    \noindent
    Because $G$ is pure, the equality of CP maps above is an equality of the corresponding linear maps up to a global phase $\alpha_x$ dependent on $x \in X$.
    Using resolution of the identity $\id{\complexs^X} = \sum_{x \in X} \ket{x} \bra{x}$ and putting all the phases together in a phase gate for the canonical basis we obtain the following equality:

    \begin{equation}
        \input{Gmeasure3.tikz}
        \hspace{5mm} = \hspace{5mm}
        \input{controlledisometry3.tikz}
        \hspace{5mm} = \hspace{5mm}
        \input{controlledisometry4.tikz}
    \end{equation}

    \noindent
    This concludes our proof.
\end{proof}

\subsection*{Proof of Proposition~\ref{proposition:coherent-control-cp-nogo}}

\setcounter{theorem_c}{\value{proposition_coherent_control0cp_nogo}}
\newcounter{eq_temp}
\setcounter{eq_temp}{\value{equation}}
\setcounter{equation}{\value{eq_coherent_control0cp_nogo}}
\begin{proposition}
    Let $(F_x)_{x \in X}$ be a family of CP maps in quantum theory.
    It is not generally possible to construct a coherent control of the family which is a function of the family $(F_x)_{x \in X}$ alone, i.e. one which is independent of a choice of purification for the CP maps in the family.
    This is the same as the statement that it is not generally possible to construct a coherent control $(\Tr_E(G'), p, m)$ of the family $(F_x)_{x \in X}$ in such a way that the following equation holds for all choices of unitaries $U_x: E \rightarrow E$:
    \begin{equation}\label{eq:coherent-control-cp-nogo:appendix}
        \input{coherentcontrolnogo1.tikz}
        \hspace{5mm} = \hspace{5mm}
        \input{coherentcontrolnogo2.tikz}
    \end{equation}
    Note that the unitaries $(U_x)_{x \in X}$ correspond to all possible choices of purification for the CP maps $(F_x)_{x \in X}$.
\end{proposition}

\setcounter{equation}{\value{eq_temp}}

\begin{proof}

    Consider a coherent control $(\Tr_E(G'), p, m)$ of a family $(F_x)_{x \in X}$ of CP maps, where $G'$ is pure.
    In particular, the following equality is satisfied:

    \begin{equation}
    	\input{V0P0measuremixed.tikz}
        \hspace{5mm} = \hspace{5mm}
        \sum_{x \in X}
        \hspace{5mm}
        \input{controlledisometrydecompositionmixed.tikz}
    \end{equation}

    \noindent
    If $\hat{F}_x$ is a chosen purification of $F_x$ for each $x \in X$---without loss of generality all with environment $E$---then we can create the coherent control of the family $(\hat{F}_x)_{x \in X}$ and obtain the following equality:

    \begin{equation}
    	\input{V0P0measuremixedi.tikz}
        \hspace{5mm} = \hspace{5mm}
        \input{xmeasuredprocess.tikz}
        \hspace{5mm} = \hspace{5mm}
        \input{controlledprocess.tikz}
    \end{equation}

    \noindent
    By essential uniqueness of purification, the following equality of pure CP maps must hold for some choice of unitary $V_x: E \rightarrow E$, dependent on each specific value of $x \in X$:

    \begin{equation}
    	\input{V0P0measuremixednodiscarding.tikz}
        \hspace{5mm} = \hspace{5mm}
        \input{controlledprocessisometry.tikz}
    \end{equation}

    \noindent
    The equality above is equivalently an equality of linear maps up to a global phase $\varphi_x$, also dependent on each specific value of $x \in X$.
    We can therefore put all $x \in X$ together and obtain the following equality of pure CP maps:

    \begin{equation}
    	\input{V0P0real.tikz}
        \hspace{5mm} = \hspace{5mm}
        \input{controlledprocessisometrynew.tikz}
    \end{equation}

    \noindent
    Any alternative choice of purification for each CP maps $F_x$ can be obtained by applying some unitary $W_x: E \rightarrow E$ to the environment of our current choice of purification $\hat{F}_x$.
    If the coherent control is to be invariant under this choice of purification, the following equality must hold for all possible choices of unitaries $(W_x)_{x \in X}$:

    \begin{equation}
    	\input{controlledprocessisometrynew2discarding.tikz}
        \hspace{5mm} = \hspace{5mm}
        \input{controlledprocessisometrynewnophasediscarding.tikz}
    \end{equation}

    \noindent
    For each $x \in X$, we can define a unitary $U_x := V_x W_x V_x^\dagger$ such that $U_x V_x = V_x W_x$, so that the equality above for all possible choices of unitaries $(W_x)_{x \in X}$ can be equivalently recast as the equality below for all possible choices of unitaries $(U_x)_{x \in X}$:

    \begin{equation}
        \input{controlledprocessisometrynew2discarding2.tikz}
        \hspace{5mm} = \hspace{5mm}
        \input{controlledprocessisometrynewnophasediscarding.tikz}
    \end{equation}

    \noindent
    The above is equivalent to Equation~\eqref{eq:coherent-control-cp-nogo:appendix} in the statement of this Proposition holding for all $(U_x)_{x \in X}$:

    \setcounter{eq_temp}{\value{equation}}
    \setcounter{equation}{\value{eq_coherent_control0cp_nogo}}

    \begin{equation}
        \input{coherentcontrolnogo1.tikz}
        \hspace{5mm} = \hspace{5mm}
        \input{coherentcontrolnogo2.tikz}
    \end{equation}

    \setcounter{equation}{\value{eq_temp}}

    \noindent
    For a general such $G'$---i.e. for a general choice of $(\hat{F}_x)_{x \in X}$---this equation cannot always be made to hold for all $(U_x)_{x \in X}$, leading to the statement of the proposition.

\end{proof}

\subsection*{Proof of Proposition~\ref{proposition:coherent-control-causal-orders}}

\setcounter{theorem_c}{\value{proposition_coherent_control0causal_orders}}
\begin{proposition}

    Let $\Phi$ be an indefinite causal scenario and $\Delta$ be a diagram over $\Phi$ in quantum theory.
    For each possible choice of phase gate $P(\varphi)$ for the computational basis of $\complexs^{\definiteCausAssoc{\Phi}}$, the coherent control of $\Delta$ with phase $P(\varphi)$ is well-defined independently of the choice of purification for the processes in $\Delta$.

\end{proposition}
\begin{proof}

    The trick to circumventing Proposition~\ref{proposition:coherent-control-cp-nogo} is the following: the choice of purification is done locally at the level of the processes associated to each event, not globally at the level of the processes associated to the diagrams.
    As a consequence, the unitary maps implementing the change of purification always factor over the individual events.
    Furthermore, the same CP map appears attached to the same event in all diagrams, marking a difference with the situation discussed in \cite{oi2003}.

    More specifically, assume that the purification $\Delta^{pure}$ used in constructing the coherent control of $\Delta$ uses purifications $\hat{F}_\omega(o_\omega | i_\omega)$ for the CP maps $F_\omega(o_\omega | i_\omega)$ appearing in $\Delta$.
    Another choice of purification corresponds to chosen unitaries $U_\omega(o_\omega | i_\omega)$, yielding as the following purifications instead:

    \begin{equation}
        \left(\id{} \otimes U_\omega(o_\omega | i_\omega)\right)
        \circ
        \hat{F}_\omega(o_\omega | i_\omega)
    \end{equation}

    \noindent
    We now fix global classical inputs and output $\underline{i} \in \underline{I}$ and $\underline{o} \in \underline{O}$ and pull all the unitaries to the boundary of the diagrams $\diagramproc{\restrict{\Delta}{\Theta}}$ for the definite scenarios $\Theta \in \definiteCausAssoc{\Phi}$.
    The same unitary $U(\underline{o} | \underline{i}): E \rightarrow E$ then appears on the global environment in all diagrams $\diagramproc{\restrict{\Delta}{\Theta}}$:

    \begin{equation}
        U(\underline{o} | \underline{i}) := \bigotimes_{\omega \in \Omega}U_\omega(o_\omega | i_\omega)
    \end{equation}

    \noindent
    The coherent control over $\complexs^{\definiteCausAssoc{\Phi}}$ is thus trivial, so that the unitary $U(\underline{o} | \underline{i})$ can be cancelled by the discarding map on the global environment $E$.
    This means that the coherent control of $\Delta$ is actually independent of the specific choice of purifications $\hat{F}_\omega$, i.e. that it is well-defined as a function of $\Delta$ alone.

\end{proof}

\end{document}